\title{A design criterion for symmetric model discrimination \\ based on nominal confidence sets}
\author{Radoslav Harman$^{\ast,\dagger}$\footnote{Department of Applied Statistics, Johannes Kepler University of Linz, Austria}\footnote{Department of Applied Mathematics and Statistics, Faculty of Mathematics, Physics and Informatics, Comenius University in Bratislava,  Slovakia}  and Werner G. M\"{u}ller$^\ast$}
\date{\today}

\documentclass[12pt]{article}

\usepackage{natbib,amsmath,amssymb,amsfonts,hyperref,graphicx,verbatim}
\usepackage{latexsym,amsmath,amsfonts,amssymb,amsthm,graphicx,epsfig}
\usepackage[lined,boxed,linesnumbered]{algorithm2e}
\SetKwRepeat{Do}{do}{while}
\def\R{\mathbb R}

\def\F{\mathbf F}

\def\a{\mathbf a}
\def\0{\mathbf 0}
\def\X{\mathfrak X}
\def\D{\mathcal D}

\newtheorem{proposition}{Proposition}

\begin{document}
\maketitle

\begin{abstract}
Experimental design applications for discriminating between models have been hampered by the assumption to know beforehand which model is the true one, which is counter to the very aim of the experiment. Previous approaches to alleviate this requirement were either symmetrizations of asymmetric techniques, or Bayesian, minimax and sequential approaches. Here we present a genuinely symmetric criterion based on a linearized distance between mean-value surfaces and the newly introduced tool of flexible nominal confidence sets. We demonstrate the computational efficiency of the approach using the proposed criterion and provide a Monte-Carlo evaluation of its discrimination performance on the basis of the likelihood ratio. An application for a pair of competing models in enzyme kinetics is given. 
\bigskip

\textbf{Keywords:} Nonlinear regression, Discrimination experiments, Exact designs, Nominal confidence sets
\end{abstract}

\section{Introduction}

Besides optimization and parameter estimation, discrimination between rival models has always been an important objective of an experiment, and, therefore, of the optimization of experimental design. The crucial problem is that one typically cannot construct an optimal model-discrimination design without already knowing which model is the true one, and what are the true values of its parameters. In this respect, the situation is analogous to the problem of optimal experimental design for parameter estimation in non-linear statistical models (e.g. 
\cite{pazman+p_14}), and many standard techniques can be used to tackle the dependence on the unknown characteristics: localization, Bayesian, minimax, and sequential approaches, as well as their various combinations.

A big leap from initial ad-hoc methods (see \cite{hill_78}  for a review), was \cite{atkinson+f_75}, who introduced $T$-optimality derived from the likelihood-ratio test under the assumption that one model is true and its parameters are fixed at nominal values chosen by the experimenter. There, maximization of the noncentrality parameter is equivalent to maximizing the power of the likelihood-ratio test for the least favourable parameter of the model, which is assumed to be wrong. Thus, $T$-optimality can be considered a combination of a localization and a minimax approach. 

When the models are nested and (partly) linear, $T$-optimality can be shown to be equivalent to $D_s$-optimality for the parameters that embody the deviations from the smaller model (see e.g. \cite{stigler_71} and \cite{dette+t_09}). For this setting the optimal design questions are essentially solved and everything hinges on the asymmetric nature of the NP-lemma. However, for a non-nested case the design problem itself is often inherently symmetric and it is the very purpose of the experiment to decide which of the two different models is true.
\bigskip

The aim of this paper is to solve the discrimination design problem in a symmetric way focussing on non-nested models. 
Thus, standard methods that are inherently asymmetric like $T$-optimality, albeit being feasible, are not a natural choice. We further suppose that we do not {use} the full prior distribution of the unknown parameters of the models, which rules out Bayesian approaches such as \cite{felsenstein_92} and \cite{tommasi+l_10}. Nevertheless, as we will make more precise in the next section, we will {utilize what can be perceived as} a specific kind of prior knowledge about the unknown parameters, extending the approach of localization. {Our goal is to provide a lean, computationally efficient and scalable method as opposed to the heavy machinery recently employed in the computational statistics literature, eg. \cite{hainy+al_18}.  Furthermore, we strive for practical simplicity, which at first }prohibits sequential (see \cite{buzzi+f_83}, \cite{mueller+p_96} and \cite{schwaab+al_06}) or sequentially generated (see \cite{vajjah+d_12}) designs.
\bigskip

A standard solution to the symmetric discrimination design problem is to employ symmetrizations of asymmetric criteria such as compound $T$-optimality, which usually depend on some weighting chosen by the experimenter. Also the minimax strategy recently presented in \cite{tommasi+al_16} is essentially a symmetrization. Moreover, usual minimax approaches lead to designs that completely depend upon the possibly unrealistic extreme values of the parameter space and their calculation again demands enormous computational effort. 
 
As the closest in spirit to our approach could be considered a proposal for linear models in Section 4.4 of \cite{atkinson+f_75} and its extension in \cite{fedorov+k_86} which, however, was not taken up by the literature. The probable reason is that it involves some rather arbitrary restrictions on the parameters as well as taking an artificial lower bound to convert it into a computationally feasible optimization problem. 
\bigskip

For expositional purposes we will now constrict ourselves to a rather specific design task but will discuss possible extensions at the end of the paper.

Let $\X \neq \emptyset$ be a finite design space and let $\D$ be a design on $\X$, i.e., a vector of design points $x_1,\ldots,x_n \in \X$, where $n$ is the chosen size of the experiment
Hence, in the terminology of the theory of optimal experimental design, we will work with \emph{exact} designs.
We will consider discrimination between a pair of non-linear regression models
\begin{eqnarray*}
  y_i&=&\eta_0(\theta_0,x_i)+\varepsilon_i, \:\: i=1,\ldots,n, \text{ and}\\
  y_i&=&\eta_1(\theta_1,x_i)+\varepsilon_i, \:\: i=1,\ldots,n,
\end{eqnarray*}
where $y_1,\ldots,y_n$ are observations, $\eta_0: \Theta_0 \times \X \to \R$, $\eta_1: \Theta_1 \times \X \to \R$ are the mean value functions, $\Theta_0  \subseteq \R^{m_0}$, $\Theta_1 \subseteq \R^{m_1}$ are parameter spaces with non-empty interiors $\mathrm{int}(\Theta_0)$, $\mathrm{int}(\Theta_1)$, and $\varepsilon_1, \ldots, \varepsilon_n$ are unobservable random errors. For both $k=0,1$ and any $x \in \X$, we will assume that the functions $\eta_k(\cdot,x)$ are differentiable on $\mathrm{int}(\Theta_k)$; the gradient of $\eta_k(\cdot,x)$ in $\theta_k \in \mathrm{int}(\Theta_k)$ will be denoted by $\nabla \eta_k(\theta_k,x)$. Our principal assumption is that one of the models is true {but we don't know which}, i.e., for $k=0$ or for $k=1$ there exists $\bar{\theta}_k \in \Theta_k$ such that $y_i=\eta_k(\bar{\theta}_k, x_i)+\epsilon_i$. 

Let the random errors be i.i.d. $N(0,\sigma^2)$, where $\sigma^2 \in (0,\infty)$. The assumption of the same variances of the errors for both models is plausible if, for instance, the errors are due to the measurement device and hence do not significantly depend on the value being measured. The situation with different error variances requires a more elaborate approach, compare with \cite{fedorov+p_68}.
\bigskip

Eventually we are aiming not just at achieving some high design efficiencies with respect to our newly proposed criterion, but want to test its usefulness in concrete discrimination experiments, that is, the probability that using our design we arrive at the correct decision about which model is the true one. So, to justify our approach numerically, we require a model discrimination rule that will be used after all observations based on the design $\D$ are collected.

The choice of the best discrimination rule based on the observations is generally a non-trivial problem. However, it is natural to compute the maximum likelihood estimates $\hat{\theta}_0$ and $\hat{\theta}_1$ of the parameters
under the assumption of the first and the second model, respectively, and then base the decision on whether 
\begin{equation}\label{eqn:rat}
  \frac{L(\hat \theta_0|(y_i)_{i=1}^n)}{L(\hat \theta_1|(y_i)_{i=1}^n)} <> 1,
\end{equation}  
i.e., the likelihood ratio being smaller or greater than $1$, or perhaps more simply whether $\log L(\hat{\theta}_0) - \log L(\hat{\theta}_1)<>0$.  Under the normality, homoskedasticity, and independence assumptions, this decision is equivalent to a decision based on the proximity of the vector $(y_i)_{i=1}^n$ of observations to the vectors of estimated mean values $(\eta_0(\hat{\theta}_0,x_i))_{i=1}^n$ and $(\eta_1(\hat{\theta}_1,x_i))_{i=1}^n$.

For the case $m_0 \neq m_1$ to counterbalance favouring models with greater number of parameters \cite{cox_13} recommends instead the use of $L(\hat{\theta}_0)/L(\hat{\theta}_1) (e^{m_1} / e^{m_0})^{n/\tilde n}$, which corresponds to the Bayesian information criterion (BIC), see \cite{schwarz_78}. Here $\tilde{n}$ corresponds to the number of observations in a real or fictitious prior experiment. 
For the sake of simplicity however, we will restrict ourselves to the case of $m:=m_0=m_1$.
Note that for the evaluational purposes we are taking a purely model selection based standpoint. More sophisticated testing procedures {for instance allowing both models to be rejected} based on the pioneering work of \cite{cox_61} are reviewed and outlined in \cite{pesaran+w_07}.
\bigskip

Let $x_1,\ldots,x_n \in \X$ and let $\D=(x_1,\ldots,x_n)$ be the design used for the collection of data prior to the decision, and assume that model $\eta_0$ is true, with the corresponding parameter value $\bar{\theta}_0$. Note that this comes without loss of generality and symmetry as we can equivalently assume model $\eta_1$ to be true. Then, the probability of the correct decision based on the likelihood ratio is equal to
\begin{equation}\label{trueprob}
  P\left[\min_{\theta_0 \in \Theta_0} \sum_{i=1}^n (\eta_0(\theta_0,x_i)-y_i))^2 \leq \min_{\theta_1 \in \Theta_1} \sum_{i=1}^n (\eta_1(\theta_1,x_i)-y_i))^2\right],
\end{equation}
where $(y_i)_{i=1}^n$ follows the normal distribution with mean $(\eta_0(\bar{\theta}_0,x_i))_{i=1}^n$ and covariance $\sigma^2 I_n$.

Clearly, probability \eqref{trueprob} depends on the true model, the unknown true parameter, and also on the unknown variance of errors. Even if these parameters were known, the probability of the correct classification would be very difficult to compute for a given design, because this requires a combination of high-dimensional integration and non-convex optimization. Therefore, it is practically impossible to directly optimize the design based on formula \eqref{trueprob}. However, we can simplify the problem by constructing a lower bound on \eqref{trueprob} which does not depend on unknown parameters and is relatively much simpler to maximize with respect to the choice of the design. The bound based on the distance $d(E_0,E_1)$, where $E_j$ is the set of all possible mean values of the observations under the model $j$, $j=0,1$, and $d$ denotes the infimum distance, is developed as follows. 

Consider a fixed experimental design $(x_1,\ldots,x_n)$, and denote $y:=(y_i)_{i=1}^n$, $\eta_j(\theta_j):=(\eta_j(\theta_j,x_i))_{i=1}^n$ for $j=0,1$. Note that we can express \eqref{trueprob} as $P[d(E_0, y) \leq d(E_1,y)]$. Now, let $R=\|\epsilon\|$, where $\epsilon=y-\eta_0(\bar{\theta}_0)$, be the norm of the vector of errors. Assuming $R \leq d(E_0,E_1)/2$ we obtain 
\begin{eqnarray*}
  d(E_0,E_1) \leq d(\eta_0(\hat{\theta}_0),\eta_1(\hat{\theta}_1)) \leq d(y,\eta_0(\hat{\theta}_0))+d(y,\eta_1(\hat{\theta}_1)) \leq \\
  d(y,\eta_0(\bar{\theta}_0))+d(y,\eta_1(\hat{\theta}_1)) = R+d(y,\eta_1(\hat{\theta}_1)) \leq d(E_0,E_1)/2 + d(y,\eta_1(\hat{\theta}_1)),
\end{eqnarray*}
which implies $d(E_0,E_1)/2 \leq d(y,\eta_1(\hat{\theta}_1))$ and consequently
\begin{equation*}
  d(E_0,y)=d(y,\eta_0(\hat{\theta}_0)) \leq d(y,\eta_0(\bar{\theta}_0))=R\leq d(E_0,E_1)/2 \leq d(y,\eta_1(\hat{\theta}_1))=d(E_1,y)
\end{equation*}
Thus, the event $[R \leq d(E_0,E_1)/2]$ implies the event $[d(E_0, y) \leq d(E_1,y)]$, that is, \eqref{trueprob} can be bounded from below by
\begin{equation}\label{lowerbound}
  P\left[R \leq d(E_0,E_1)/2\right].
\end{equation}
To make \eqref{trueprob} as high as possible, it makes sense to maximize \eqref{lowerbound}, i.e., maximize $d(E_0,E_1)$, which depends on the underlying experimental design. 
{
While this maximization is much simpler than maximizing \eqref{trueprob} directly, it still generally requires non-convex multidimensional optimization at each iteration of the maximization procedure, which is impractical for computing exact optimal designs.  A realistic approach must be numerically feasible and circumvent the problems of the dependence of the design on unknown true model parameters, which we will achieve by rapidly computable approximation of $d(E_0,E_1)$ through linearization, as will be  explained in the following section.
}

\bigskip

\subsection*{A motivating example}
Let $\eta_0(\theta_0,x) = \theta_0 x$ and $\eta_1(\theta_1,x) = e^{\theta_1 x}$. Furthermore for the moment we assume just two observations $y_1, y_2$ at fixed design points $x_1=-1$ and $x_2=1$, respectively. In this case evidently $\hat{\theta}_0 = \frac{y_2-y_1}{2}$ and $\hat{\theta}_1$ is the solution of $2 e^{-\theta } \left(y_1-e^{-\theta }\right)-2 e^{\theta } \left(y_2-e^{\theta }\right)=0$, which for $-2 \le y_1 \le 2$ is the root of the polynomial $\theta^4-\theta^3 y_2 + \theta y_1 -1$. Figure \ref{fig1} displays the loglikelihoodratio contours for the original and linearized models and it is obvious that the former are non-convex and complex while the latter are much simpler, convex, and do approximate fairly well. Note that whilst this example is for a fixed design it motivates why the linearizations can serve as the cornerstones of our design method as will become clearer in the following sections.  

\begin{figure}[h] 
\includegraphics[width=0.5\textwidth]{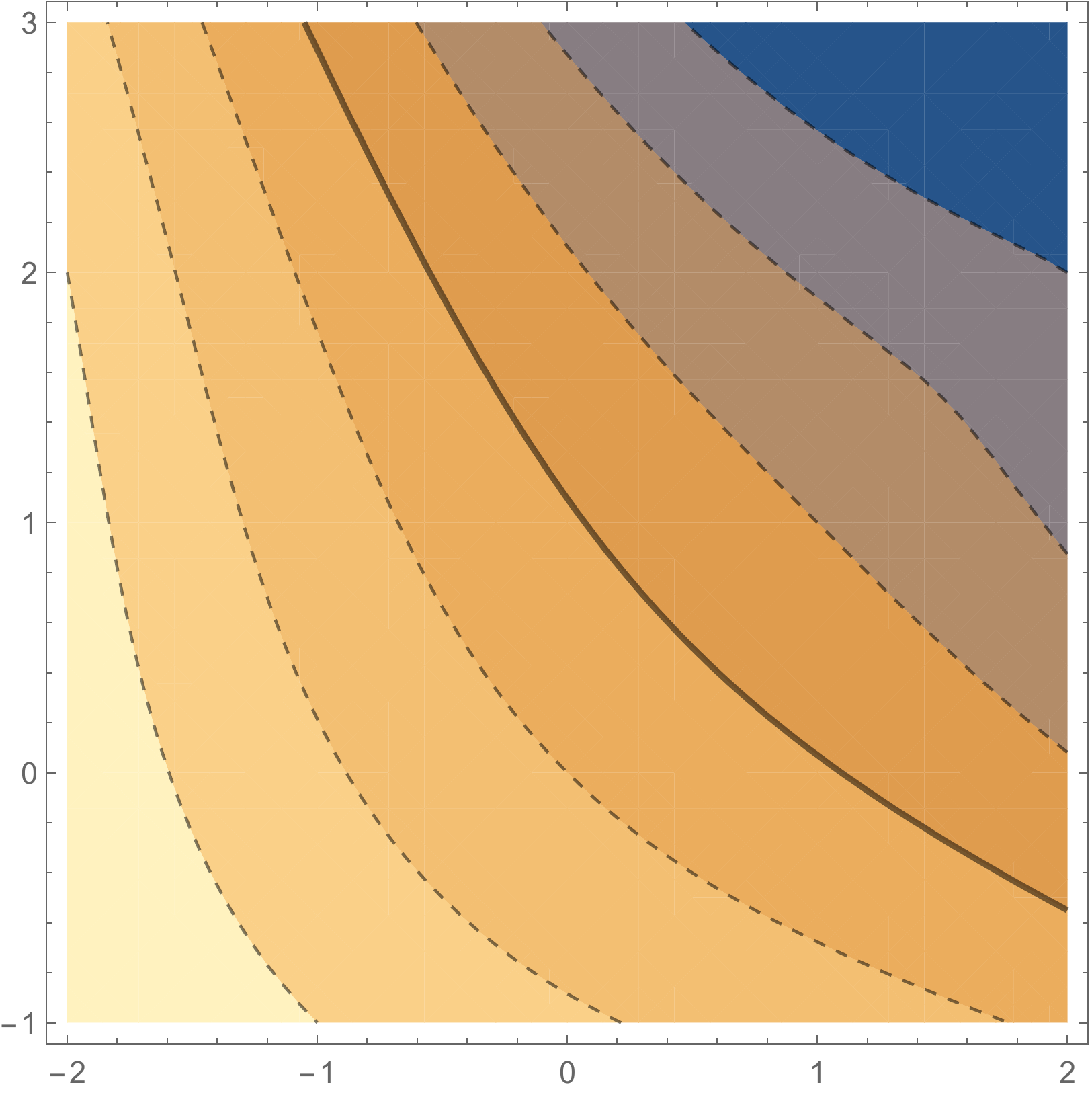}
\includegraphics[width=0.5\textwidth]{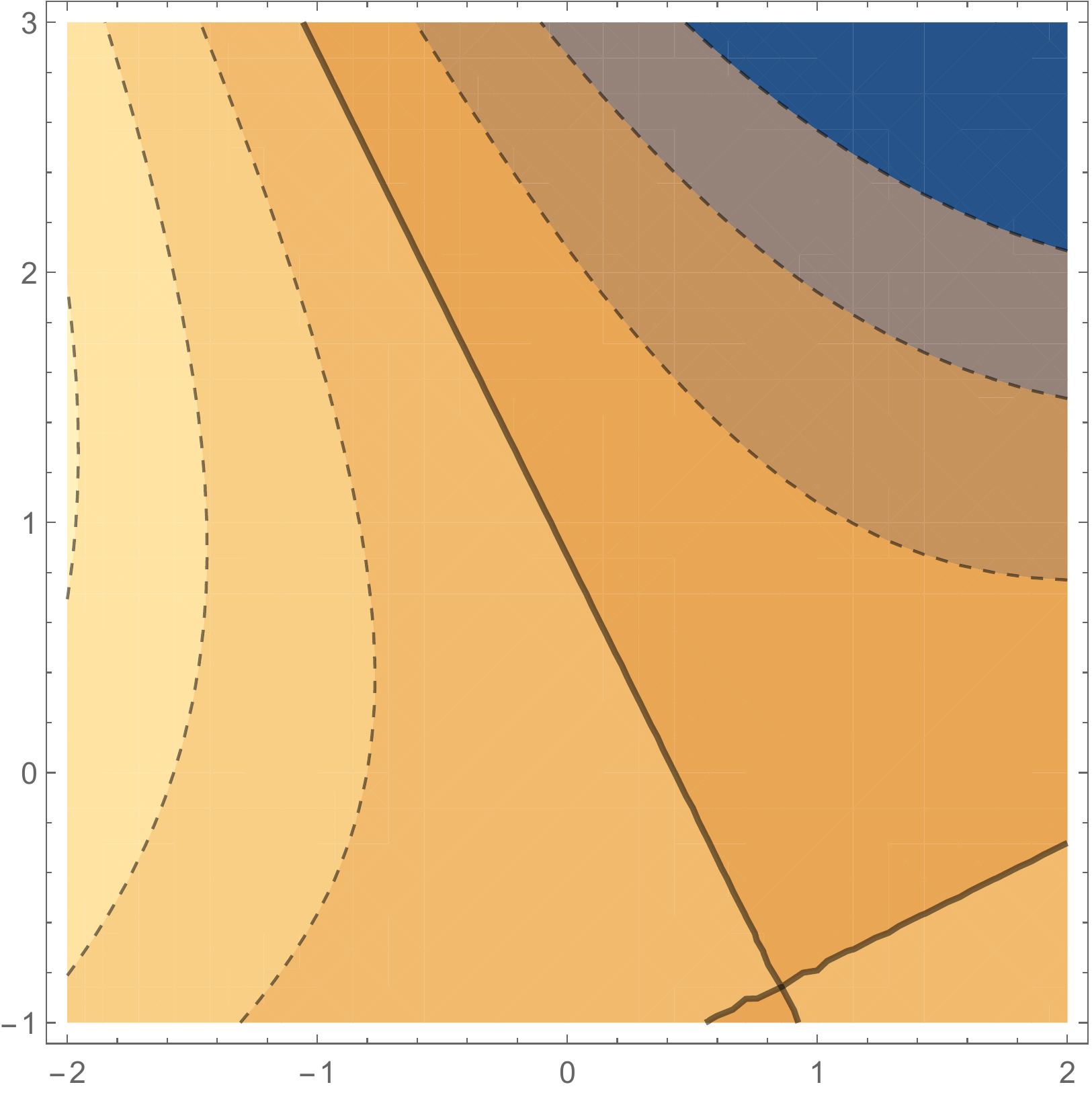}
\caption{left panel: contour plot of $\log L(\hat{\theta}_0) - \log L(\hat{\theta}_1)$ for Example 1, solid line corresponds to $0$; right panel:  corresponding contour plot for the model $\eta_1$ linearized at $\theta_1=1$.}
\label{fig1}
\end{figure}

\section{The linearized distance criterion}

We suggest an extension of the idea of localization used for the non-linear experimental design. Let $\tilde\theta_0 \in \mathrm{int}(\Theta_0)$ and $\tilde\theta_1 \in \mathrm{int}(\Theta_1)$ be nominal parameter values, which satisfy the basic \emph{discriminability condition} $\eta_0(\tilde\theta_0,x) \neq \eta_1(\tilde\theta_1,x)$ for some $x \in \X$. Let us introduce regions $\tilde\Theta_0 \subseteq \mathrm{int}(\Theta_0) \subseteq \mathbb{R}^m$ and $\tilde\Theta_1 \subseteq \mathrm{int}(\Theta_1) \subseteq \mathbb{R}^m$ containing $\tilde\theta_0$ and $\tilde\theta_1$; we will consequently call $\tilde\Theta_0$ and $\tilde\Theta_1$   {\em nominal confidence sets}. {It is evident that optimal designs depend upon the parameter spaces in the same way as on our nominal confidence sets (cf. \cite{dette+al_13}), but the latter will not be considered fixed like the parameter spaces $\Theta_0$ and $\Theta_1$, and can thus be used as a tuning device for our procedure, which has not been done before.}
\bigskip

Let $\D=(x_1,\ldots,x_n)$ be a design. Let us perform the following particular linearization of Model 
$\eta_{k=0,1}$ in $\tilde{\theta}_k$:
\begin{equation*}
  (y_i)_{i=1}^n \approx \F_k(\D)\theta_k+\a_k(\D)+\varepsilon,
\end{equation*}
where $\F_k(\D)$ is an $n \times m$ matrix given by
\begin{equation*}
  \F_k(\D)=\left(\nabla \eta_k(\tilde{\theta}_k, x_1), \ldots, \nabla \eta_k(\tilde{\theta}_k, x_n)\right)^T,
\end{equation*}
$\a_k(\D)$ is an $n$-dimensional vector
\begin{equation*}
  \a_k(\D)=(\eta_k(\tilde{\theta}_k,x_i))_{i=1}^n-\F_k(\D)\tilde{\theta}_k,
\end{equation*}
and $\varepsilon=(\varepsilon_1,\ldots,\varepsilon_n)^T$ is a vector of independent $N(0,\sigma^2)$ errors.

Note that for the proposed method the vector $\a_k(\D)$ plays an important role and, although it is known, we cannot subtract it from the vector of observations, as is usual when we linearize a single non-linear regression model. However, if $\eta_k$ corresponds to the standard linear model then $\a_k(\D)=\0_n$ for any $\D$.

\subsection{Definition of the $\delta$ criterion}

Let $\D$ be a design. Consider the design criterion
\begin{eqnarray}
\delta(\D)&=&\inf_{\theta_0 \in \tilde\Theta_0, \theta_1 \in \tilde\Theta_1} \delta(\D |\theta_0,\theta_1), \text{ where }\label{D1}\\
 \delta(\D |\theta_0,\theta_1) &=& \left\| \a_0(\D) + \F_0(\D)\theta_0-\{\a_1(\D) +  \F_1(\D)\theta_1\}\right\|, \label{D}
\end{eqnarray} 
for $\theta_0 \in \tilde\Theta_0, \theta_1 \in \tilde\Theta_1$. The criterion $\delta$ can be viewed as an approximation of the nearest distance of the mean-value surfaces of the models, in the neighbourhoods of the vectors $(\eta_0(\tilde{\theta}_0,x_i))_{i=1}^n$ and $(\eta_1(\tilde{\theta}_1,x_i))_{i=1}^n$; 
see the illustrative Figure \ref{ill}.

\begin{figure}[htb] 
\centering{
\includegraphics[width=\textwidth,height=0.4\textheight]{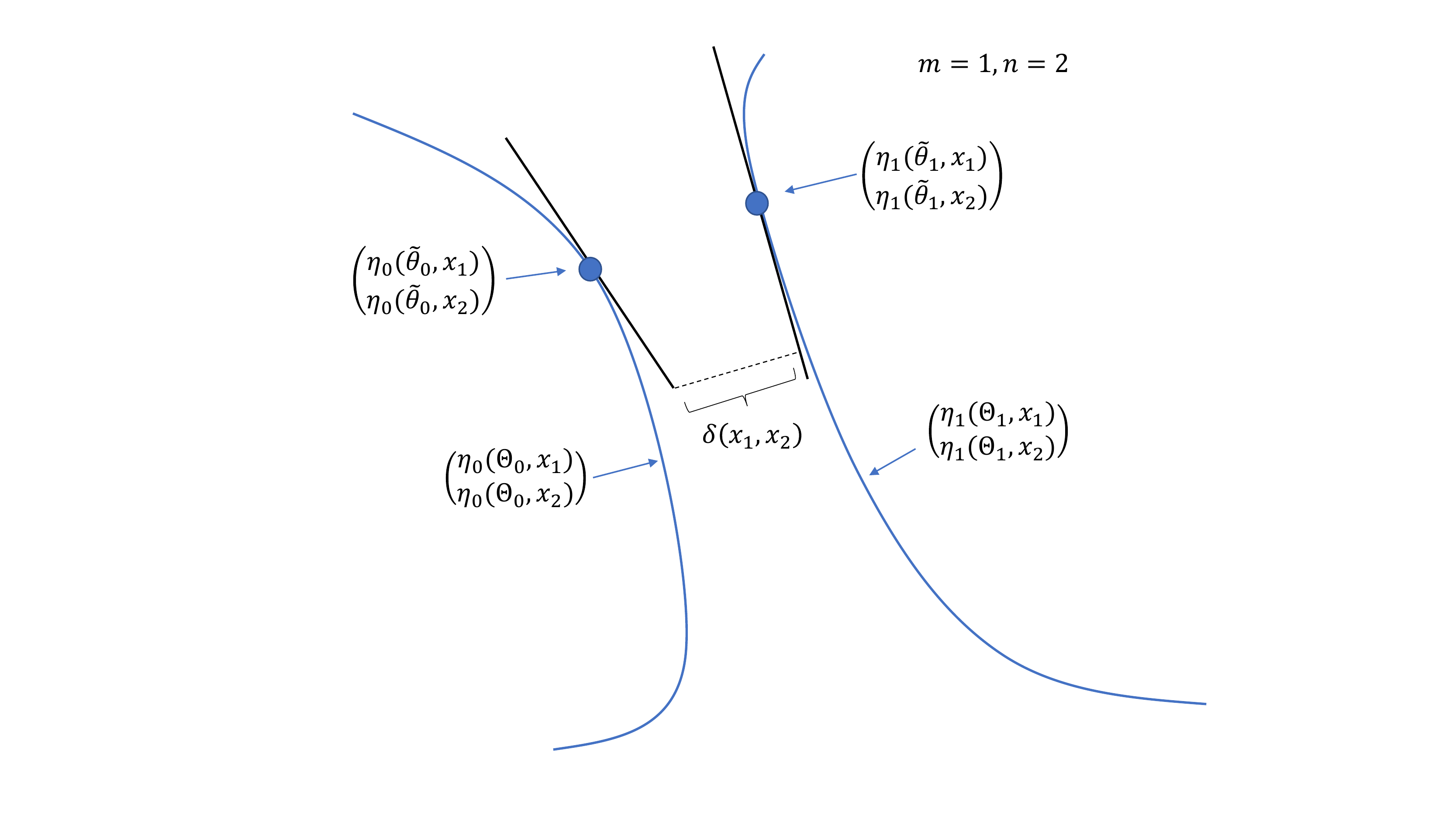}}
\vskip-0.5cm
\caption{Illustrative graph for the definition of $\delta(\D)$ for a one-parametric model ($\Theta_0, \Theta_1 \subseteq \R$) and a design of size two ($\D=(x_1,x_2)$). The line segments correspond to the sets $\{\a_0(\D)+\F_0(\D)\theta_0: \theta_0 \in \tilde{\Theta}_0\}$ and $\{\a_1(\D)+\F_1(\D)\theta_1: \theta_1 \in \tilde{\Theta}_1\}$ for some nominal confidence {intervals} $\tilde{\Theta}_0$ and $\tilde{\Theta}_1$.} 
\label{ill}
\end{figure}

We will now express the $\delta$-criterion as a function of the design $\D=(x_1,\ldots,x_n)^T$ represented by a measure $\xi$ on $\X$ defined as  
\begin{equation*}
  \xi(\{x\}):=\#\big\{i \in \{1,\ldots,n\}: x_i=x\big\}, \: x \in \X, 
\end{equation*}
where $\#$ means the size of a set. Let $\tilde{\theta}=(\tilde{\theta}_0^T, \tilde{\theta}_1^T)^T$. For all $x \in \X$ let
\begin{eqnarray*}
 \Delta \eta(\tilde{\theta}, x)&:=&\eta_0(\tilde{\theta}_0, x) - \eta_1(\tilde{\theta}_1, x), \\
 \nabla \eta(\tilde{\theta}, x)&:=&\left(\nabla \eta_0^T(\tilde{\theta}_0, x), \: -\nabla \eta_1^T(\tilde{\theta}_1, x)\right)^T.
\end{eqnarray*}
For any $\theta_0 \in \tilde\Theta_0$, $\theta_1 \in \tilde\Theta_1$ and $\theta=(\theta_0^T, \theta_1^T)^T$ we have
\begin{eqnarray}
\delta^2(\D|\theta_0,\theta_1) &=& \left\|\a_0(\D) + \F_0(\D)\theta_0-\{\a_1(\D) +  \F_1(\D)\theta_1\}\right\|^2 \nonumber \\
&=& \sum_{i=1}^n \left(\nabla \eta^T(\tilde{\theta}, x_i)(\theta-\tilde{\theta}) + \Delta \eta(\tilde{\theta}, x_i)\right)^2 \nonumber \\
&=& \int_{\X} \left(\nabla \eta^T(\tilde{\theta}, x)(\theta-\tilde{\theta}) + \Delta \eta(\tilde{\theta}, x)\right)^2  \mathrm{d}\xi(x). \label{eqn:deltaInt}
\end{eqnarray}
Therefore
\begin{equation}\label{eq:xi}
\delta^2(\D|\theta_0,\theta_1) = (\theta-\tilde{\theta})^T \mathbf{M}(\xi,\tilde{\theta}) (\theta-\tilde{\theta}) + 2 \mathbf{b}^T(\xi,\tilde{\theta})(\theta-\tilde{\theta}) + c(\xi,\tilde{\theta}), 
\end{equation}
where
\begin{eqnarray}
\mathbf{M}(\xi,\tilde{\theta})&=& \int_{\X} \nabla\eta(\tilde{\theta}, x) \nabla \eta^T(\tilde{\theta}, x) \mathrm{d}\xi(x), \label{eqn:M}\\
\mathbf{b}(\xi,\tilde{\theta})&=&  \int_{\X} \Delta\eta(\tilde{\theta}, x) \nabla\eta(\tilde{\theta}, x) \mathrm{d}\xi(x), \label{eqn:b}\\
c(\xi,\tilde{\theta}) &=& \int_{\X} [\Delta\eta(\tilde{\theta}, x)]^2 \mathrm{d}\xi(x). \label{eqn:c}
\end{eqnarray}
The matrix $\mathbf{M}(\xi,\tilde{\theta})$ in equations \eqref{eq:xi} and \eqref{eqn:M} can be recognized as the information matrix for the parameter $\theta$ in the linear regression model
\begin{eqnarray}
  z_i&=&\nabla \eta^T(\tilde{\theta},x_i)\theta + \epsilon_i \nonumber \\
     &=&[\F_0(\D), -\F_1(\D)]_{i\cdot} \theta + \epsilon_i; \: i=1,\ldots,n, \label{eqn:DRM}
\end{eqnarray}
where $[\F_0(\D), -\F_1(\D)]_{i\cdot}$ is the $i$-th row of the matrix $[\F_0(\D), -\F_1(\D)]$, with parameter $\theta$ and independent, homoskedastic errors $\epsilon_1,\ldots,\epsilon_n$ with mean $0$; we will call \eqref{eqn:DRM} a \emph{response difference model}.

\subsection{Computation of the $\delta$ criterion value for a fixed design}

For a fixed design $\D$, expression \eqref{D} shows that $\delta^2(\D|\theta)$ is a quadratic function of $\theta=(\theta_0^T,\theta_1^T)^T$. Moreover, both $\delta(\D|\theta)$ and $\delta^2(\D|\theta)$ are convex, because they are compositions of an affine function of $\theta$ and convex functions $\|.\|$ and $\|.\|^2$, respectively. Clearly, if the nominal confidence sets are compact, convex and polyhedral, optimization \eqref{D1} can be efficiently performed by specialized solvers for linearly constrained quadratic programming.
\bigskip

Alternatively, we can view the computation of $\delta(\D|\theta)$ as follows. Since
\begin{equation*}
 \delta^2(\D|\theta_0,\theta_1) = \left\|\{\a_0(\D)-\a_1(\D)\} -[-\F_0(\D), \F_1(\D)] \theta\right\|^2,
\end{equation*}
the minimization in \eqref{D1} is equivalent to computing the minimum sum of squares for a least squares estimate of $\theta$ restricted to $\tilde{\Theta}:=\tilde{\Theta}_0 \times \tilde{\Theta}_1$ in the response difference model with artificial observations 
\begin{equation*}
\tilde{z}_i=\{\a_0(\D)-\a_1(\D)\}_i, \: i=1,\ldots,n.
\end{equation*}

Thus, if $\tilde{\Theta}_0=\tilde{\Theta}_1=\R^m$, the infimum in \eqref{D1} is attained, and it can be computed using the standard formulas of linear regression in the response difference model. If the nominal confidence sets are compact cuboids, \eqref{D1} can be evaluated by the very rapid and stable method for bounded variables least squares implemented in the \texttt{R} package \texttt{bvls}; see \cite{stark+p_95} and \cite{mullen_13}.
\bigskip

The following simple proposition collects the analytic properties of a natural analogue of $\delta$ defined on the linear vector space $\Xi$ of all finite signed measures on $\X$.  

\begin{proposition}\label{approx}
For $\theta_0 \in \tilde\Theta_0$, $\theta_1 \in \tilde\Theta_1$ and a  finite signed measure $\xi$ on $\X$ let $\delta^2_{app}(\xi|\theta_0,\theta_1)$ be defined via formula \eqref{eqn:deltaInt}. Then, $\delta^2_{app}(\cdot|\theta_0,\theta_1)$ is linear on $\Xi$. Moreover, let
\begin{equation*}
  \delta^2_{app}(\xi):=\inf_{\theta_0 \in \tilde\Theta_0, \theta_1 \in \tilde\Theta_1} \delta^2_{app}(\xi |\theta_0,\theta_1).
\end{equation*}
Then, $\delta^2_{app}$ is positive homogeneous and concave on $\Xi$.
\end{proposition}

Positive homogeneity of $\delta^2_{app}$ implies that an $s$-fold replication of an exact design leads to an $s$-fold increase of its $\delta^2$ value. Consequently, a natural and statistically interpretable definition of relative $\delta$-efficiency of two designs $\D_1$ and $\D_2$ is given by $\delta^2(\D_1)/\delta^2(\D_2)$, provided that $\delta^2(\D_2)>0$.
\bigskip 

Let $\mathfrak{D}$ be the set of all $n$-point designs. A design $\D^* \in \mathfrak{D}$ will be called $\delta$-optimal, if 
\begin{equation*}
  \D^* \in \mathrm{argmax}_{\D \in \mathfrak{D}} \delta(\D).
\end{equation*}
Note that the basic discriminability condition implies that if $\tilde\Theta_0=\{\tilde\theta_0\}$ and $\tilde\Theta_1=\{\tilde\theta_1\}$, then $\delta(\D^*)$ is strictly positive. However, for larger nominal confidence sets it can happen that $\delta(\D^*)=0$.

As the evaluation of the $\delta$-criterion is generally very rapid, a $\delta$-optimal design, or a nearly $\delta$-optimal design can be computed similarly as for the standard design criteria. For instance, in small problems we can use complete-enumeration and in larger problems we can employ an exchange heuristic, such as the KL exchange algorithm (see e.g. \cite{atkinson+al_07}).
\bigskip

Note that the $\delta$-optimal designs depend not only on $\eta_0$, $\eta_1$, $\X$, $n$, $\tilde{\theta}_0$ and $\tilde{\theta}_1$, but also on the nominal confidence sets $\tilde{\Theta}_0$ and $\tilde{\Theta}_1$. 


\subsection{Parametrization of nominal confidence sets}
\bigskip

For simplicity, we will focus on cuboid nominal confidence sets centered at the nominal parameter values. This choice can be justified by the results of \cite{sidak_67}, in particular if we already have confidence intervals for individual parameters{, see further discussion in Section \ref{Sec:outlook}}. Specifically, we will employ the homogeneous dilations
\begin{equation}\label{rsystem}
\tilde\Theta_{k}^{(r)} := r\left(\tilde\Theta_{k}^{(1)} - \tilde\theta_{k}\right) + \tilde\theta_{k}, \qquad r \in [0,\infty), \: k=0,1, 
\end{equation}
$\tilde\Theta_{0}^{(\infty)}:=\mathbb{R}^m$, $\tilde\Theta_{1}^{(\infty)}:=\mathbb{R}^m$, such that $r$ can be considered a tuning parameter governing the size of the nominal confidence sets. In \eqref{rsystem}, $\tilde\Theta_{0}^{(1)}$ and $\tilde\Theta_{1}^{(1)}$ are ``unit'' non-degenerate compact cuboid confidence sets centred in respective nominal parameters. For any design $\D$ and $r \in [0,\infty]$, we define
\begin{equation}\label{delr}
 \delta_r(\D):=\inf_{\theta_0 \in \tilde\Theta_{0}^{(r)}, \theta_1 \in \tilde\Theta_{1}^{(r)}} \delta(\D |\theta_0,\theta_1).
\end{equation}
Note that for our choice of nominal confidence sets the infimum in \eqref{delr} is attained. The $\delta_r$-optimal values of the problem will be denoted by
\begin{equation*}
 o(r):=\max_{\D \in \mathfrak{D}} \delta_r(\D).
\end{equation*}

\begin{proposition}\label{funr}
a) Let $\D$ be a design. Functions $\delta^2_r(\D)$, $\delta_r(\D)$, $o^2(r)$, $o(r)$ are non-increasing and convex in $r$ on the entire interval $[0,\infty]$. b) There exists $r^* < \infty$, such that for all $r \geq r^*$: (i) $o(r) = o(\infty)$; (ii) Any $\delta_{\infty}$-optimal design is also a $\delta_r$-optimal design.
\end{proposition}
\begin{proof}
 a) Let $\D$ be an $n$-point design and let $0 \leq r_1 \leq r_2 \in [0,\infty]$.
 
 Inequality $\delta^2_{r_1}(\D) \geq \delta^2_{r_2}(\D)$ follows from definitions \eqref{rsystem}, \eqref{delr}, and inequality $o^2(r_1) \geq o^2(r_2)$ follows from the fact that a maximum of non-increasing functions is a non-increasing function. Monotonicity of $\delta_r(\D)$ and $o(r)$ in $r$ can be shown analogously.
 
 To prove the convexity of $\delta^2_r(\D)$ in $r$, let $\alpha \in (0,1)$ and let $r_{\alpha}=\alpha r_1 + (1-\alpha) r_2$. For all $r \in [0,\infty]$, let $\hat{\theta}_r$ denote a minimizer of $\delta^2_r(\D|\cdot)$ on $\tilde\Theta^{(r)} := \tilde\Theta^{(r)}_0 \times \tilde\Theta^{(r)}_1$. Convexity of $\delta^2(\D|\theta)$ in $\theta$ and a simple fact $\alpha \hat{\theta}_{r_1} + (1-\alpha) \hat{\theta}_{r_2} \in \tilde\Theta^{(r_{\alpha})}$ yield
\begin{eqnarray*}
  \alpha \delta^2_{r_1}(\D)+(1-\alpha)\delta^2_{r_2}(\D) = \alpha \delta^2(\D|\hat{\theta}_{r_1})+(1-\alpha)\delta^2(\D|\hat{\theta}_{r_2}) \\
  \geq \delta^2(\D|\alpha \hat{\theta}_{r_1}+(1-\alpha)\hat{\theta}_{r_2}) \geq \delta^2(\D|\hat{\theta}_{r_{\alpha}})= \delta^2_{r_{\alpha}}(\D),
\end{eqnarray*}
which proves that $\delta^2_r(\D)$ is convex in $r$. The convexity of $\delta_r(\D)$ in $r$ can be shown analogously. The functions $o^2$ and $o$, as point-wise maxima of a system of convex functions, are also convex.

b) For any design $\D$ of size $n$, the function $\delta^2_\infty(\D|\cdot)$ is non-negative and quadratic on $\mathbb{R}^{2m}$, therefore its minimum is attained in some $\theta_{\D} \in \mathbb{R}^{2m}$. There is only a finite number of exact designs of size $n$, and $\tilde\Theta^{(r)} \uparrow_r \mathbb{R}^{2m}$, which means that there exists $r^* < \infty$ such that $\theta_{\D} \in \tilde\Theta^{(r^*)}$ for all designs $\D$ of size $n$. Let $r \geq r^*$. We have
\begin{equation*}
  o(\infty)=\max_{D \in \mathfrak{D}} \min_{\theta \in \mathbb{R}^{2m}} \delta_\infty(\D|\theta)= \max_{D \in \mathfrak{D}} \min_{\theta \in \tilde\Theta^{(r)}} \delta(\D|\theta) = \max_{D \in \mathfrak{D}} \delta_{r}(\D|\theta)=o(r),
\end{equation*}
proving (i). Let $\D^{(\infty)}$ be any $\delta_\infty$-optimal $n$-trial design. The equality (i) and the fact that $\delta_r(\D^{(\infty)})$ and $o(r)$ are non-increasing with respect to $r$ gives 
\begin{equation*}
  \delta_r(\D^{(\infty)}) \geq \delta_\infty(\D^{(\infty)}) = o(\infty) = o(r^*) \geq o(r),
\end{equation*}
Which proves (ii).  
\end{proof}

The second part of Proposition \ref{funr} implies the existence of a finite interval $[0,r^*]$ of relevant confidence parameters; increasing the confidence parameter beyond $r^*$ keeps the set of optimal designs as well as the optimal value of the $\delta$-criterion unchanged. We will call any such $r^*$ an \emph{upper confidence bound}. 
\bigskip

Algorithm \ref{rmax} provides a simple iterative method of computing $r^*$. Our experience shows that it usually requires only a small number of re-computations of the $\delta_r$-optimal design, even if $r_{ini}$ is small and $q$ is close to $1$, resulting in a good upper confidence bound $r^*$ (see the meta-code of Algorithm \ref{rmax} for details). 

\begin{algorithm}
	\SetKwInOut{Input}{Input}\SetKwInOut{Output}{Output}
	\Input{Pre-computed value $o(\infty)$, an initial confidence $r_{ini}>0$, a ratio $q>1$}
	\Output{An upper confidence bound $r^*$} 
	\BlankLine 
	Set $r \leftarrow r_{ini}$ and $fin \leftarrow 0$\\
	Compute a $\delta_r$-optimal design, denote it by $\D$\\
	\If {$\delta_r(\D)=o(\infty)$} {
	   Set $fin \leftarrow 1$
	}
	\While {$fin=0$} {
	  Set $r \leftarrow q.r$\\
	  \If {$\delta_r(\D) \leq o(\infty)$} {
		Recompute a $\delta_r$-optimal design, denote it by $\D$\\
		\If {$\delta_r(\D) \leq o(\infty)$} {
		  Set $fin \leftarrow 1$
		}
      }
    }
    Set $r^* \leftarrow r$
	\caption{A simple algorithm for computing an upper confidence bound. Due to the high speed and stability of the computation of the values of $\delta_r$ for candidate designs, it is possible to use an adaptation of the standard KL exchange heuristic to compute the input value $o(\infty)$, as well as to obtain $\delta_r$-optimal designs in steps 2 and 9 of the algorithm itself.}\label{rmax} 
\end{algorithm}

\subsection*{The motivating example continued}
Consider the models from the motivating example. Let $\X=\{1.00, 1.01, \ldots, 2.00\}$, $\tilde{\theta}_0=e$, and $\tilde{\theta}_1=1$. Note that these nominal values satisfy $\eta_0(\tilde\theta_0,1)=\eta_1(\tilde\theta_1,1)$. Moreover, let us set $\tilde\Theta^{(0)}=[e-1,e+1]$ and $\tilde\Theta^{(1)}=[0,2]$, and let the required size of the experiment be $n=6$. First, we computed the value $o(\infty) \approx 0.02614$. Next, we used Algorithm \ref{rmax} with $r_{ini}=0.3$ and $q=1+10^{-6}$, which returned an upper confidence bound $r^* \approx 0.6787$ after as few as $7$ computations of $\delta_r$-optimal designs. Informed by $r^*$, we computed $\delta_r$-optimal designs for $r=0.01, 0.1, 0.2, \ldots, 0.7$. The resulting $\delta_r$-optimal designs are displayed in Figure \ref{tmec}. Note that if $\tilde\Theta^{(r)}$'s are very narrow, the $\delta_r$-optimal design is concentrated in the design point $x=2$, effectively maximizing the difference between $\eta_0(\tilde\theta_0,x)$ and $\eta_1(\tilde\theta_1,x)$. For larger values of $r$, the $\delta_r$-optimal design has a $2$-point and ultimately a $3$-point support.
\bigskip

\begin{figure}[ht] 
\vskip-1.cm
\centering
\includegraphics[width=0.7\textwidth]{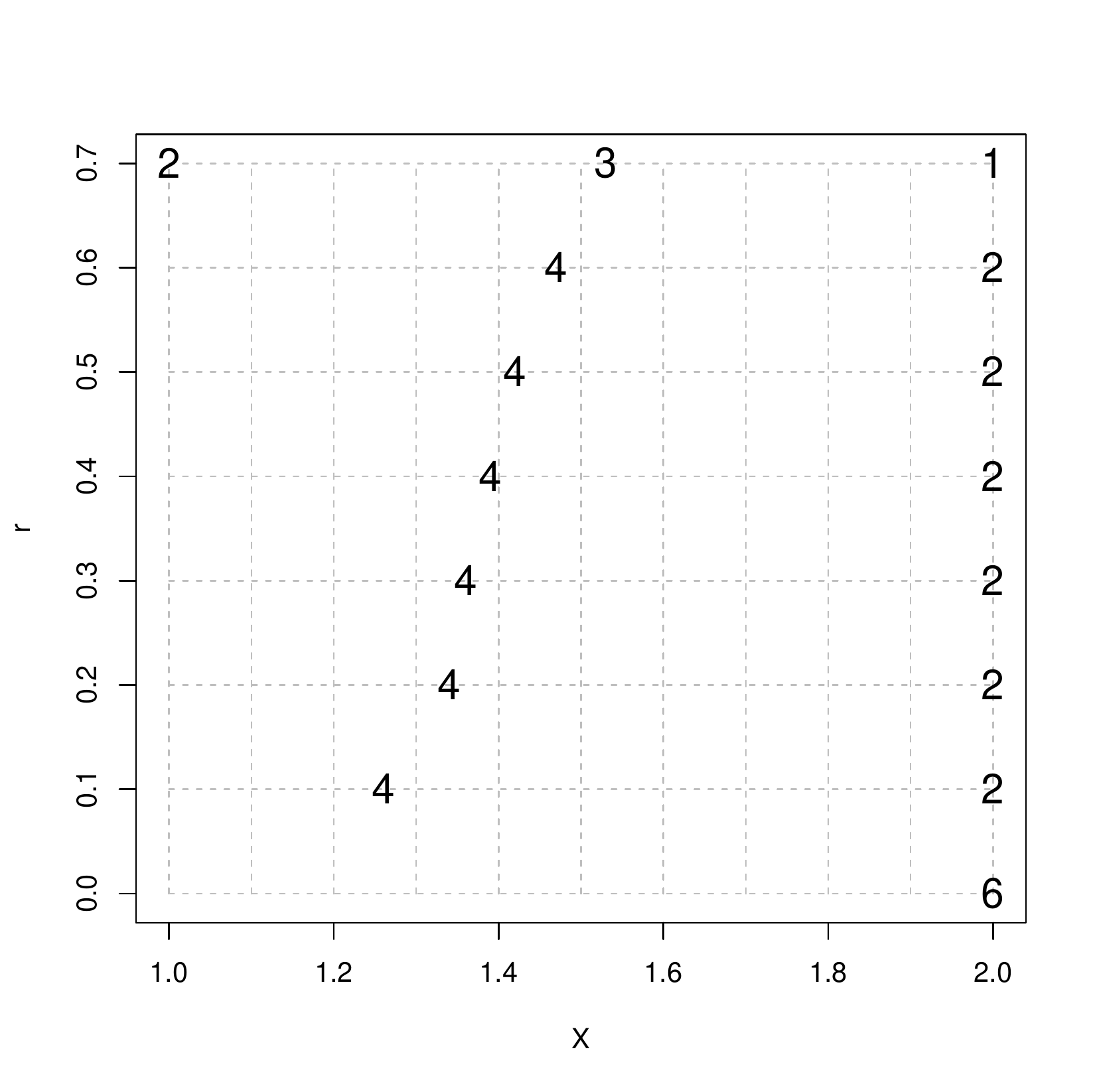}
\vskip-0.5cm
\caption{$\delta_r$-optimal designs of size $n=6$ for different $r$'s; see the second part of the motivating example. The horizontal axis corresponds to the design space, and the vertical axis corresponds to different spans $r$ of the nominal confidence sets. For each $r$, the figure displays the number of repeated observations at different design points, corresponding to the $\delta_r$-optimal design.}\label{tmec}
\end{figure}

For some pairs of competing models there exists an upper confidence bound $r^*$, beyond which the values of $\delta_r$ are constantly $0$ for all designs. These cases can be identified by solving a linear programming (LP) problem, as we show next.

\begin{proposition}\label{ECB}
Let $\bar\D$ be the design which performs exactly one trial in each point of $\X$. Consider the following LP problem with variables $r \in \mathbb{R}$, $\theta_0 \in \mathbb{R}^m$, $\theta_1 \in \mathbb{R}^m$:
\begin{eqnarray}\label{LP}
\min && r\\
\mathrm{s.t.} && \F_0(\bar\D)\theta_0+\a_0(\bar\D) = \F_1(\bar\D)\theta_1+\a_1(\bar\D), \nonumber \\
&& \theta_0 \in \tilde\Theta^{(r)}_0, \: \theta_1 \in \tilde\Theta^{(r)}_1, \: r \geq 0.\nonumber 
\end{eqnarray}
 Assume that \eqref{LP} has some solution, and denote one solution of \eqref{LP} by $(r^*,\theta_a^T,\theta_b^T)^T$. Then, $r^*$ is a finite upper confidence bound. Moreover, $o(r)=0$ for all $r \in [r^*,\infty]$. 
\end{proposition}
\begin{proof}
From the expression \eqref{eq:xi} we see that for any design $\D$ and its non-replication version $\D^{nr}$ we have: $\delta_r(\D^{nr})=0$ implies $\delta_r(\D)=0$. Moreover, if $\D_2 \succeq \D_1$ in the sense that $\D_2$ is an augmentation of $\D_1$ then: $\delta_r(\D_2)=0$ implies $\delta_r(\D_1)=0$. Now let $(r^*,\theta_a^T,\theta_b^T)^T$ be a solution of \eqref{LP}, let $r \geq r^*$ and let $\D$ be any design. Definition of $\delta_r$ and the form of \eqref{LP} imply $\delta_r(\bar\D)=0$. From $\bar{\D} \succeq \D^{nr}$ we see that then $\delta_r(\D^{nr})=0$, hence $\delta_r(\D)=0$. The proposition follows.
\end{proof}

Note that $r^*$ obtained using Proposition \ref{LP} does not depend on $n$, i.e., it is an upper confidence bound simultaneously valid for all design sizes. The basic discriminability condition implies that $r^* \neq 0$.
\bigskip

If the competing models are linear, vectors $\a_0(\bar\D)$ and $\a_1(\bar\D)$ are zero. Therefore, \eqref{LP} has a feasible solution $(r,\0_m^T,\0_m^T)^T$ for any $r \geq 0$ such that both $\tilde\Theta^{(r)}_0$ and $\tilde\Theta^{(r)}_1$ cover $\0_m$. That is, for the case of linear models, there is a finite upper confidence bound $r^*$ beyond which the $\delta_r$-values of all designs vanish. However, the same holds for specific non-linear models, including the ones from Section \ref{Sec:Ex2}:      
\begin{proposition}\label{condlin}
 Assume that both competing regression models are linear provided that we consider a proper subset of their parameters as known constants. Then \eqref{LP} has a finite feasible solution, i.e., there exists a finite upper confidence bound $r^*$ such that $o(r)=0$ for all $r \in [r^*,\infty]$.
\end{proposition}
\begin{proof}
Without loss of generality, assume that fixing the first $k_0<m$ components of $\theta_0$ converts Model 0 to a linear model. More precisely, let $\theta_{01},\ldots,\theta_{0m}$ denote the components of $\theta_0$ and assume that
\begin{equation*}
\eta_0(\theta_0,x)=\sum_{j=k_0+1}^m \gamma^{(0)}_j(\theta_{01},\ldots,\theta_{0k_0},x)\theta_{0j}
\end{equation*}
for some functions $\gamma^{(0)}_j$, $j=k_0+1,\ldots,m$. Choose $\hat{\theta}_0$ such that $\hat\theta_{0j}=\tilde\theta_{0j}$ for $j=1,\ldots,k_0$, and $\hat\theta_{0j}=0$ for $j=k_0+1,\ldots,m$. Make an analogous assumption for Model 1 and also define $\hat{\theta}_1$ analogously. It is then straightforward to verify that for the design $\bar{\D}$ from Proposition \ref{ECB} we have $\F_k(\bar\D)\hat\theta_k+\a_k(\bar\D)=\0_d$, where $d=\#\X$, for both $k=0,1$. Therefore, any $(r,\hat\theta_0^T,\hat\theta_1^T)^T$ such that $\hat\theta_0 \in \tilde\Theta^{(r)}_0$ and $\hat\theta_1 \in \tilde\Theta^{(r)}_1$ is a solution of \eqref{LP}.
\end{proof}

In the following we numerically demonstrate that the $\delta$ design criterion leads to designs which yield a high probability of correct discrimination. 

\section{An application in enzyme kinetics} \label{Sec:Ex2} 

This real applied example is taken from \cite{bogacka+al_11} and was already used in \cite{atkinson_12} to illustrate model discrimination designs. There two types of enzyme kinetic reactions are considered, where the reactions velocity $y$ is alternatively modeled as
\begin{equation} \label{competitive}
y=\frac{\theta_{01}x_1}{\theta_{02}\left(1+\frac{x_2}{\theta_{03}}\right)+x_1}+\epsilon,
\end{equation}
and
\begin{equation} \label{noncompetitive}
y=\frac{\theta_{11}x_1}{(\theta_{12}+x_1)\left(1+\frac{x_2}{\theta_{13}}\right)}+\epsilon,
\end{equation}
which represent competitive and noncompetitive inhibition, respectively. 
Here $x_1$ denotes the concentration of the substrate and $x_2$ the concentration of an inhibitor. The data used in \cite{bogacka+al_11} is on Dextrometorphan-Sertraline and yields the estimates displayed in Table \ref{table1}. Assumed parameter spaces were not explicitely given there, but can be inferred from their figures as $\theta_{0,1},\theta_{1,1} \in (0,\infty)$, $\theta_{0,2},\theta_{1,2} \in (0,60]$, and $\theta_{0,3},\theta_{1,3} \in (0,30]$, respectively. Designs for parameter estimation in these models were recently given in \cite{schorning+al_17}.

\begin{table}[h]
\centering
\begin{tabular}{|l|l|l|c|l|l|l|}
\hline
       & estimate $\hat\theta$ & st.err. $\hat \sigma_\theta$ &  &      & estimate $\hat\theta$ & st.err. $\hat\sigma_\theta$                   \\
$\theta_{01}$ &    7.298      &    0.114 &   & $\theta_{11}$ &    8.696      &     0.222                \\
$\theta_{02}$ &    4.386      &    0.233 &   & $\theta_{12}$ &    8.066      &     0.488                 \\
$\theta_{03}$ &    2.582      &    0.145 &   & $\theta_{13}$ &   12.057      &     0.671\\
\hline          
\end{tabular}
\caption{Parameter estimates and corresponding standard errors for models (\ref{competitive}) and (\ref{noncompetitive}), respectively.}
\label{table1}
\end{table}

In \cite{atkinson_12} the two models are combined into an encompassing model
\begin{equation} \label{combined}
y=\frac{\theta_{21}x_1}{\theta_{22}\left(1+\frac{x_2}{\theta_{23}}\right)+x_1\left(1+\frac{(1-\lambda)x_2}{\theta_{23}}\right)}+\epsilon,
\end{equation}
where $\lambda=1$ corresponds to (\ref{competitive}) and $\lambda=0$ to (\ref{noncompetitive}), respectively. Following the ideas of \cite{atkinson_72} as used e.g. in \cite{atkinson_08} or \cite{perrone+al_17} one can then proceed to find so-called $D_s$-optimal {(i.e. D-optimal for only a subset of parameters)} designs for $\lambda$ and employ them for model discrimination. Note that also this method is not fully symmetric as it requires a nominal value for $\lambda$ for linearization of (\ref{combined}), which induces some kind of weighting.

The nominal values used in \cite{atkinson_12} obviously motivated by the estimates of (\ref{competitive}) were $\tilde \theta_{01} = \tilde\theta_{11} = \tilde \theta_{21} =10$, $\tilde\theta_{02} = \tilde\theta_{12} = \tilde \theta_{22}=4.36$, $\tilde\theta_{03}=2.58$, $\tilde\theta_{13}=5.16$, and $\tilde \theta_{23} = 3.096$. However, note that particularly for model (\ref{noncompetitive}) the estimates in Table \ref{table1} give considerably different values and also nonlinear least squares directly on (\ref{combined}) yields the deviating estimates given in Table \ref{table2}. The design region used was rectangular ${\cal X} = {\cal X}_1 \times {\cal X}_2 = [0,30] \times [0,40]$.

\begin{table}[h]
\centering
\begin{tabular}{|l|l|l|}
\hline
       &  estimate $\hat\theta$ & st.err. $\hat \sigma_\theta$                  \\
$\theta_{21}$ &    7.425      &    0.130                \\
$\theta_{22}$ &    4.681      &    0.272                 \\
$\theta_{23}$ &    3.058      &    0.281                 \\
$\lambda$ 	& 0.964      &    0.019 \\
\hline         
\end{tabular}
\caption{Parameter estimates and corresponding standard errors for the encompassing model (\ref{combined}).}
\label{table2}
\end{table}

In table 2 of \cite{atkinson_12} four approximate optimal designs (we will denote them A1-A4) were presented: the $T-$optimal designs assuming $\lambda=0$ (A1) and $\lambda=1$ (A4), a compound $T$-optimal design (A3) and a $D_s$-optimum (A2) for the encompassing model (for the latter note that Atkinson assumed $\lambda=0.8$ whereas the estimate suggest a much higher value). We will compare our $\delta$-optimal designs against properly rounded (by the method of \cite{pukelsheim+r_92}) exact versions of these designs.  

\subsection{Confirmatory experiment $n=6$, normal errors}

Let us first assume we want to complement the knowledge from our initial experiment by another experiment for which, however, we were given only limited resources, e.g. for the sample sizes of mere $n=6$ observations. Note that the aim is not to augment the previous 120 observations but to make a confirmatory decision just out of the new observations. That is we are using the data from the initial experiment just to provide us with nominal values for parameter estimates and noise variances for the simulation respectively. {This is a realistic scenario if for instance for legal reasons the original data had to be deleted and only summary information was kept available}.

As we are assuming equal variances for the two models we are using the estimate for the error standard deviation $\hat \sigma = 0.1526$ from the encompassing model as a base value for the simulation error standard deviation. However, using $\hat \sigma$ was not very revealing for the hit rates were consistently high for all designs. Thus to accentuate the differences the actual standard deviation used was $2 \times \hat \sigma$ instead (unfortunately an even higher inflation is not feasible as it would result in frequent negative observations leading to faulty ML-estimates). We then simulated the data generating process under each model for $N=10000$ times and calculated the total percentages of correct discrimination (hit rates) when using the likelihood ratio as decision rule. 

We are comparing the designs A1-A4 to three specific delta designs $\delta1, \delta2$, and $\delta3$ which represent a range of different nominal intervals. Specifically we chose  $\tilde\Theta_k=[\tilde\theta_{k1} \pm r\tilde\sigma_{k1}] \times [\tilde\theta_{k2} \pm r\tilde\sigma_{k2}] \times [\tilde\theta_{k3} \pm r\tilde\sigma_{k3}]_{k=0,1}$, where we chose $\tilde\theta_{kj} = \hat\theta_{kj}$ and $\tilde\sigma_{kj}=\hat\sigma_{kj}$ for $k=0,1$ and $j=1,2,3$. The tuning parameter $r$ was set to three levels: $r=1$ (which is close to the lower bound of still providing a regular design), $r=5$ and $r=15$ (which is sufficiently close to the theoretical upper bound to yield a stable design), respectively.  To make the latter more precise: the models in considerations are such that if we fix the last two out of the three parameters, then they become one-parametric linear models. Therefore, using Proposition \ref{condlin} we know that there exists a finite upper confidence bound $r^*$. Solving  \eqref{LP} provides the numerical value $r^* \approx 64.02$. Note that the same bound is valid for all design sizes $n$. While A1-A4 and $\delta$1 all contain 4 support points, while $\delta$2 has 6 and $\delta$3 5, respectively. A graphical depiction of the designs is given in Figure \ref{designs}. 

\begin{figure}[ht] 
\vskip-0.5cm
\centering
\includegraphics[width=0.24\textwidth]{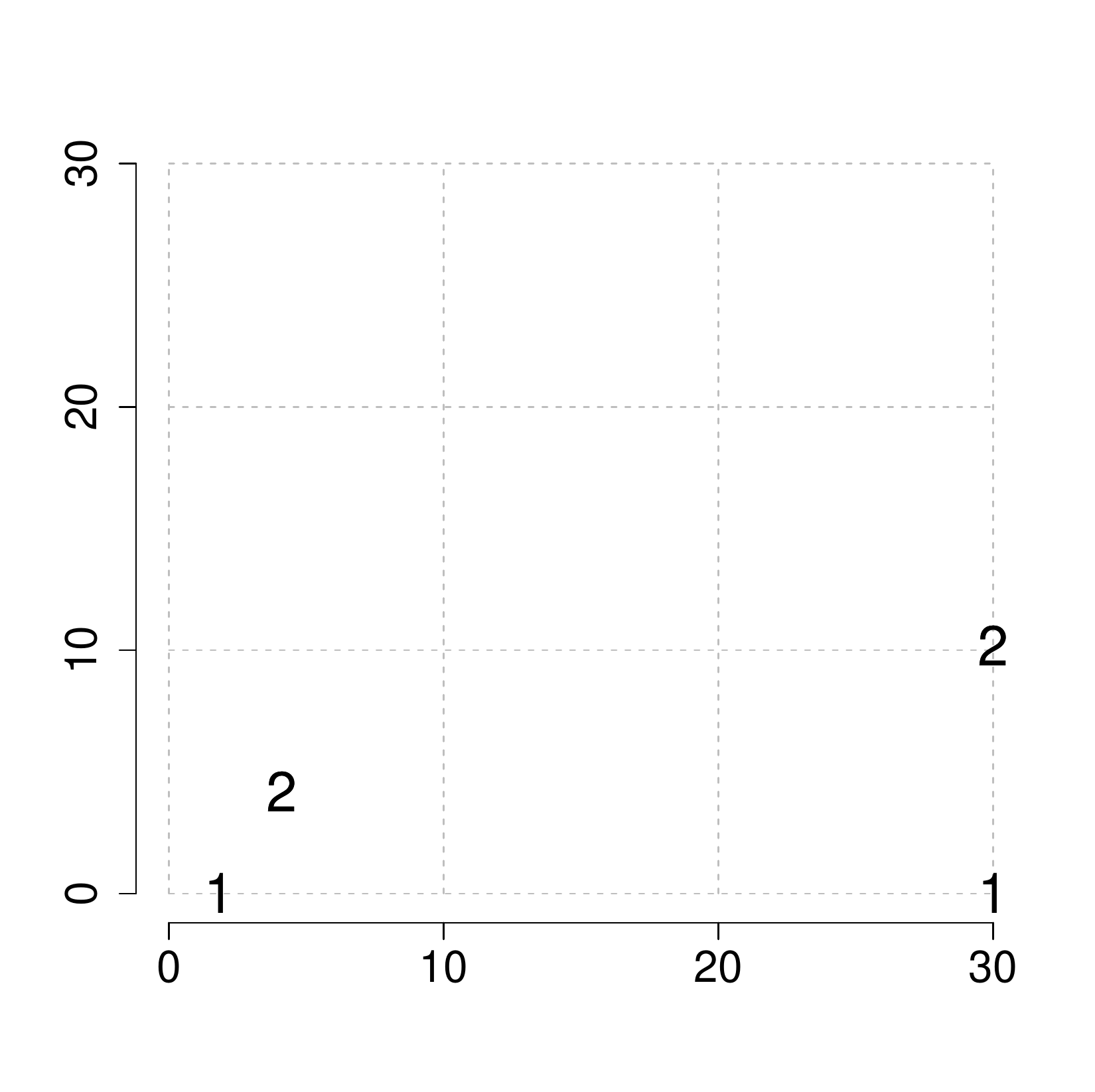}
\includegraphics[width=0.24\textwidth]{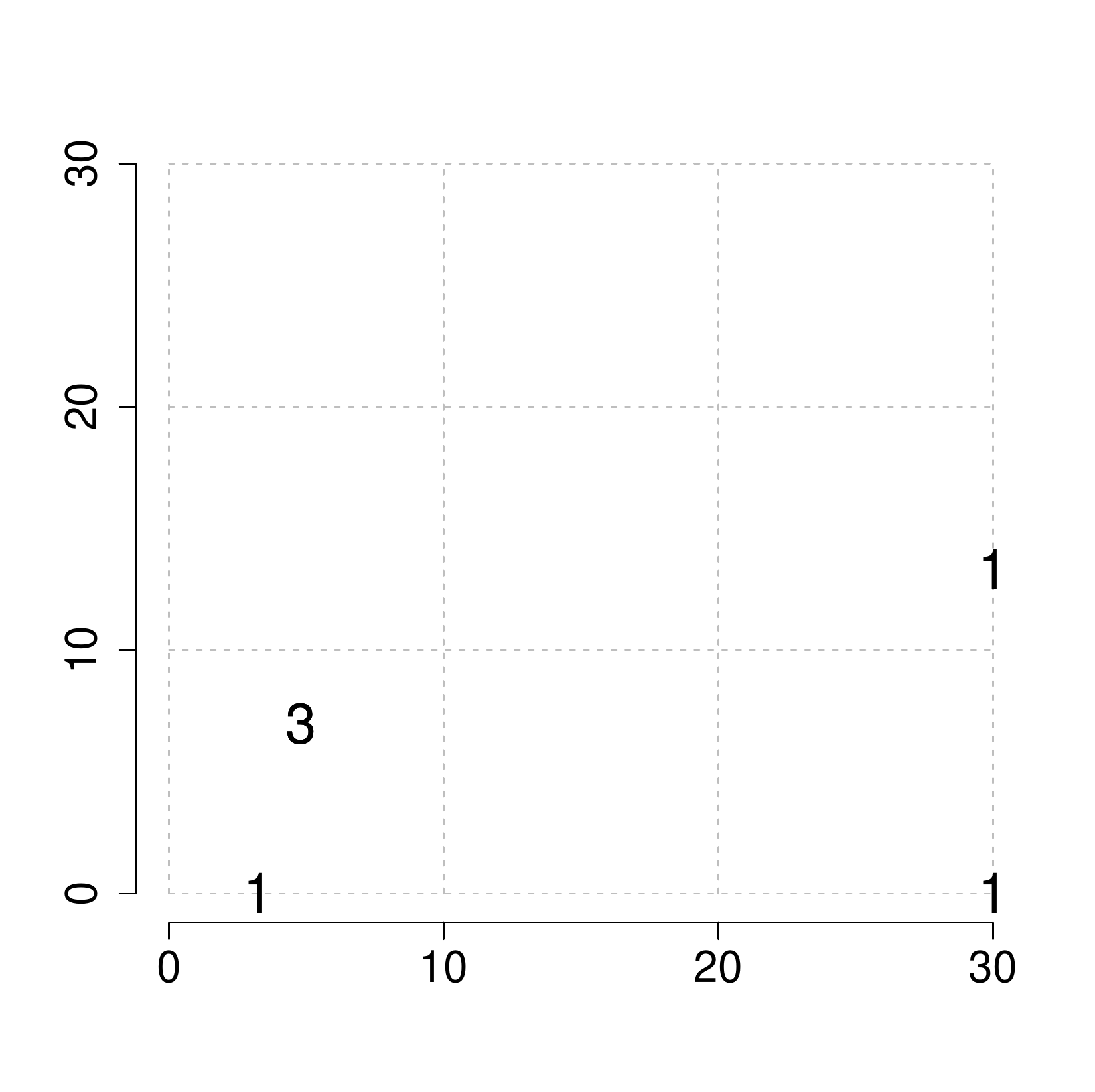}
\includegraphics[width=0.24\textwidth]{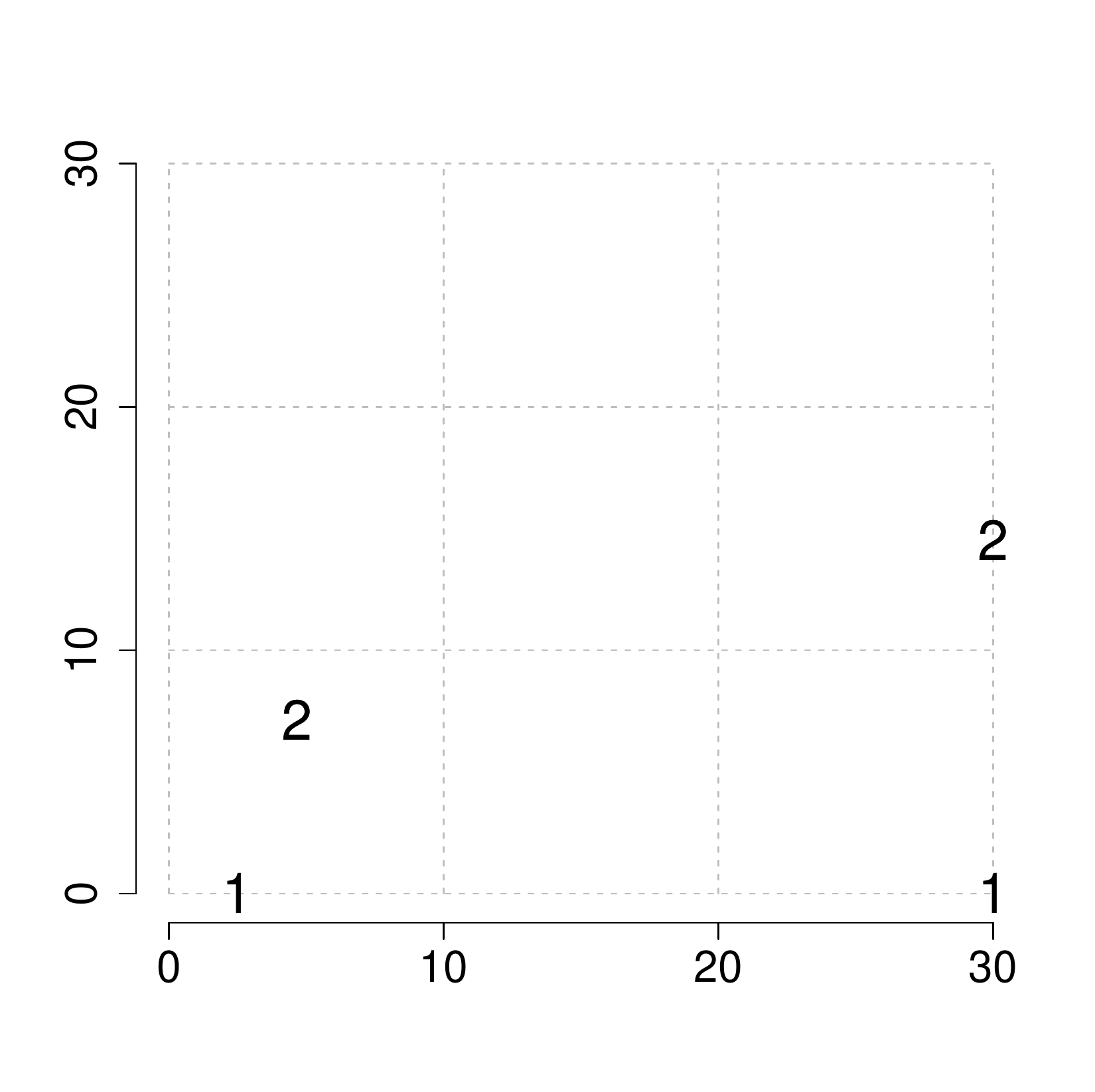}
\includegraphics[width=0.24\textwidth]{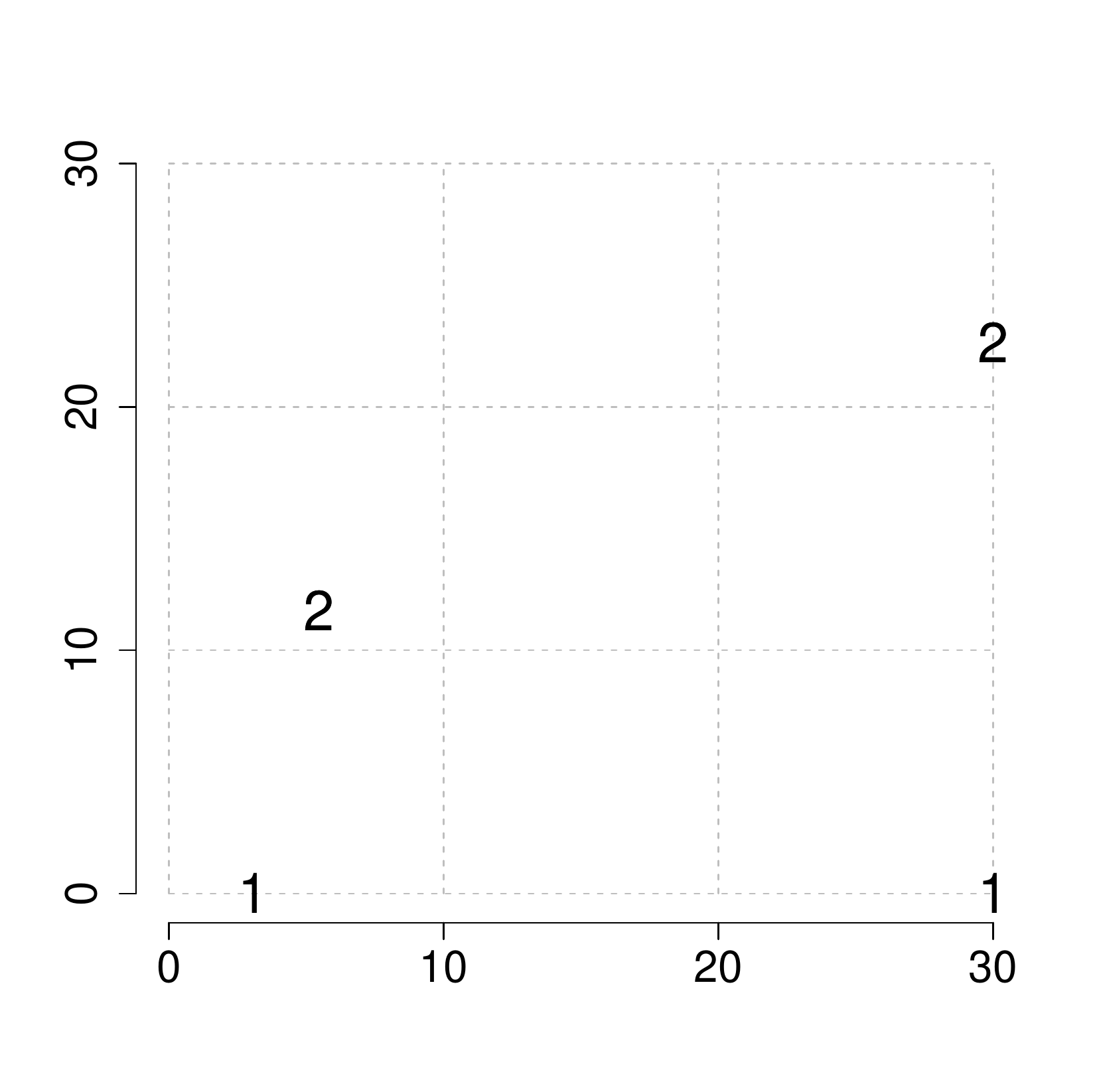}
\includegraphics[width=0.24\textwidth]{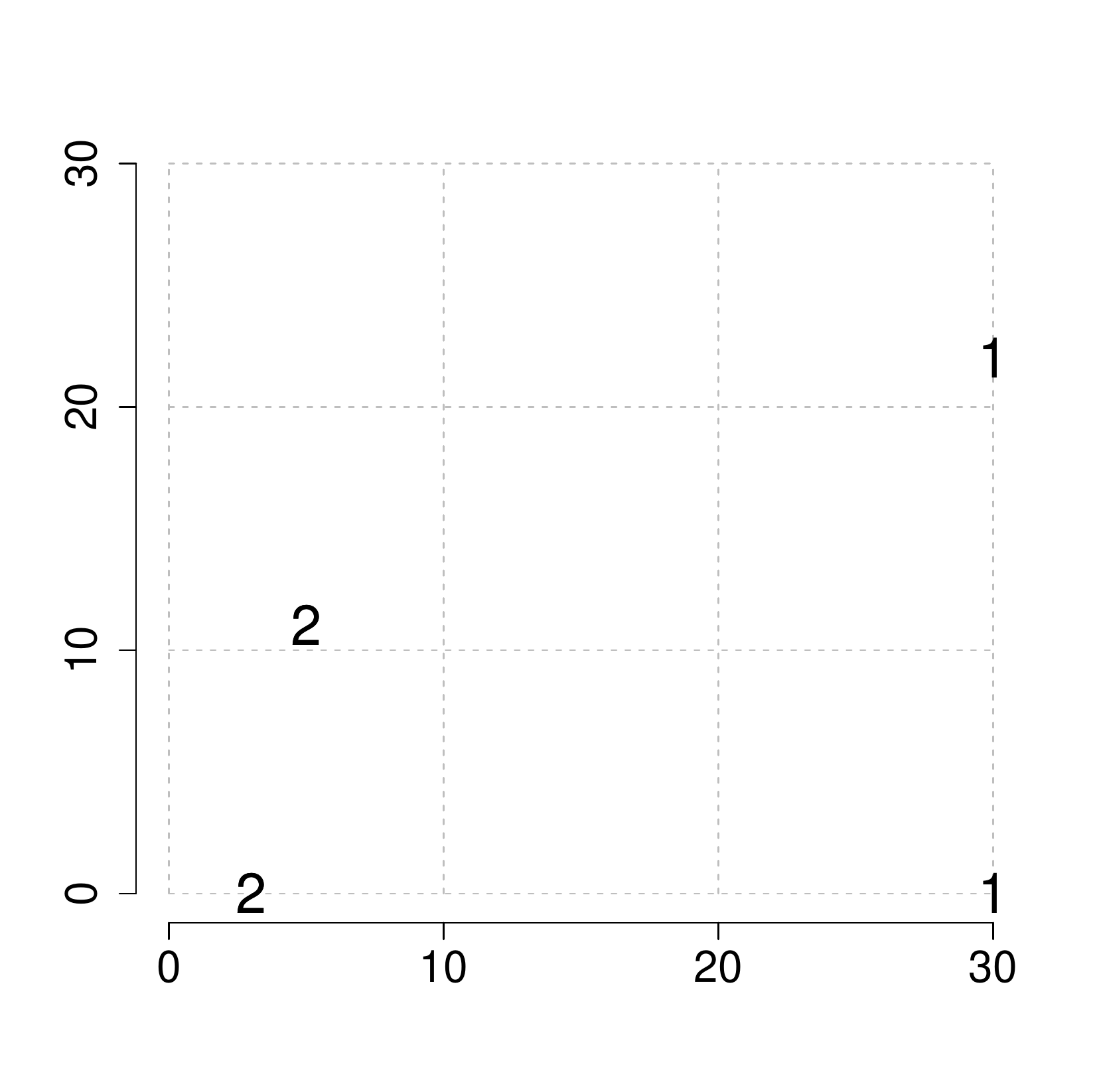}
\includegraphics[width=0.24\textwidth]{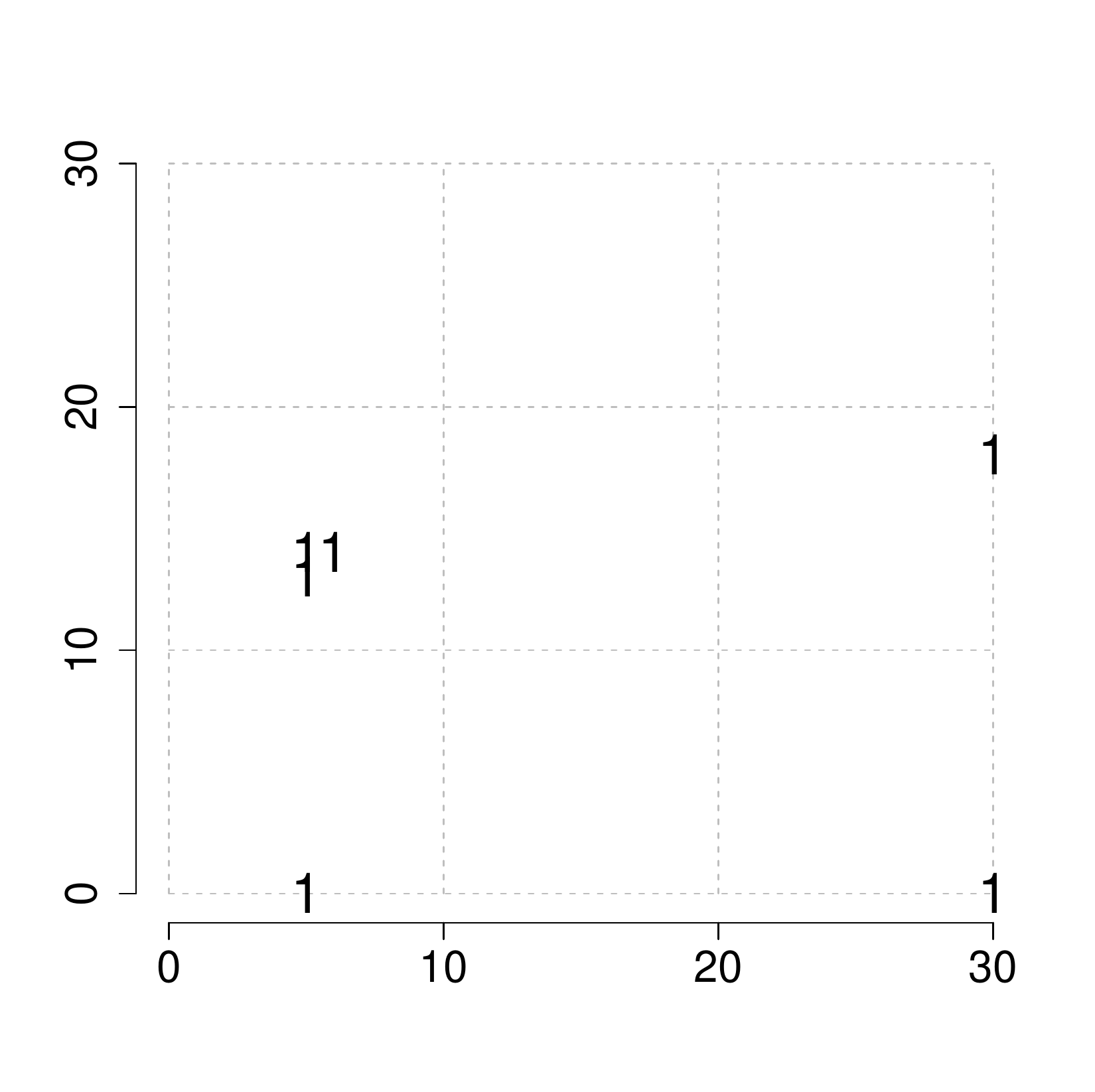}
\includegraphics[width=0.24\textwidth]{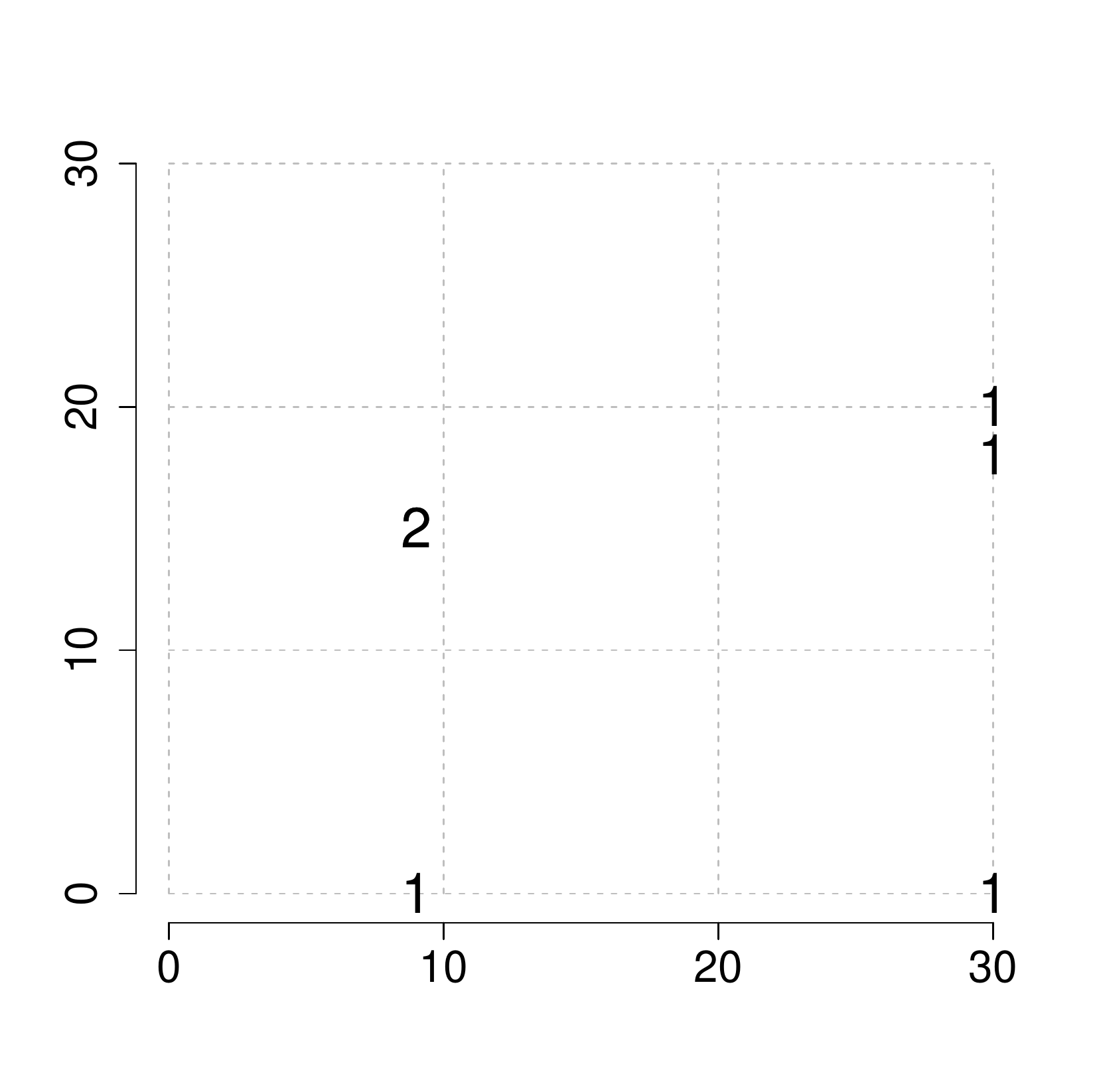}
\vskip-0.5cm
\caption{Compared designs: first row A1-A4, second row $\delta$1-$\delta$3.}\label{designs}
\end{figure}

{\noindent\textbf{Robustness study:}
As we would like to avoid to compare designs only if the data is generated from the nominal values (although this favours all designs equally) we} perturbed the data generating process by drawing parameters from uniform distributions drawn at $\tilde \theta \pm c \times \tilde\sigma_\theta$, where $c$ then acts as a pertubation parameter. Under these settings all these designs fare pretty well as can be seen from  Table \ref{table3}.  However, $A4$ and $\delta2$ seem to outperform the other competing designs by usually narrow margins except perhaps for $A1$, which is consistently doing worst. Note that in a real situation the true competitors of $\delta$-optimal designs are just $A2$ and $A3$ as it is unknown beforehand which model is true.

\begin{table}[h]
\centering
\begin{tabular}{|c|c|c|c|c|c|c|}
\hline
c            & \multicolumn{2}{c|}{0} & \multicolumn{2}{c|}{1} & \multicolumn{2}{c|}{5} \\ \hline
true model   & $\eta_0$    & $\eta_1$   & $\eta_0$    & $\eta_1$   & $\eta_0$    & $\eta_1$   \\ \hline
A1            & 91.11      & 94.45     & 91.35      & 93.95     & 90.44      & 93.24     \\ \hline
A2            & 97.11      & 96.75     & 97.47      & 96.64     & 96.74      & 96.27     \\ \hline
A3            & 96.60      & 96.51     & 96.47      & 96.40     & 95.69      & 96.06     \\ \hline
A4            & {\bf 97.94}      & 96.57     & 97.73      & 96.29     & 97.62      & 96.07     \\ \hline
$\delta1$ & 97.59      & 95.11     & 97.43      & 94.90     & {\bf 97.71}      & 94.56     \\ \hline
$\delta2$ & 97.93      & {\bf 97.03}     & {\bf 97.77}      & {\bf 96.67}     & 97.20      & {\bf 96.54}     \\ \hline
$\delta3$ & 96.50      & 95.29     & 96.42      & 95.36     & 96.19      & 95.64     \\ \hline
\end{tabular}
\caption{Total hit rates for $N=10000$ under each model.}
\label{table3}
\end{table}
 
\subsection{A second large scale experiment $n=60$, lognormal errors}

We would like to investigate the respective pereformance in a larger scale setting, where potential rounding effects are neglibile. 
For that purpose, using additive normal errors in the data generating process turns out unfeasible as the discriminatory power of all the designs for $n=60$ is nearly perfect without inflating error variance. Inflating the variance by a large enough factor, however, would generate a large number of negative observations, which renders likelihood estimation invalid. So, the data generating process was adapted to use multiplicative lognormal errors. The observations were then rescaled to match the means from the original process. This way we are ad liberty to inflate the error variance by any factor without producing faulty observations. Note that now the data generating process does not fully match the assumptions under which the designs were generated, but this can just be considered an extended robustness study as it holds for all compared designs equally. {We could of course also have calculated the designs under the same data-generating process, but as the fit of the model to the original data is not greatly improved and models (\ref{competitive}) and (\ref{noncompetitive}) seem firmly established in the parmacological  literature, we refrained from doing this.}

Perturbation of the parameters here did not exhibit a discernible effect, while the error inflation still does. For brevity we here report only again the results for using  $5 \times \hat \sigma$ (and $c=0$). The respective designs $\delta$1-3 were qualitatively similar to those given in Figure \ref{designs} albeit with more diverse weights. In this simulation we generated 100 instances of $n=60$ observations from these designs a thousand times.

The corresponding boxplots of the correct classification rates are given in Figure \ref{boxplot}.   In this setting A4 seems a bit superior even under $\eta_1$ (remember it being the $T$-optimum design assuming $\eta_0$ true), while $\delta$1 and $\delta$2 come close (and beat the true competitors A2 an A3) with A1 again being clearly the worst. 

\begin{figure}[ht] 
\vskip-0.5cm
\centering
\includegraphics[width=0.7\textwidth]{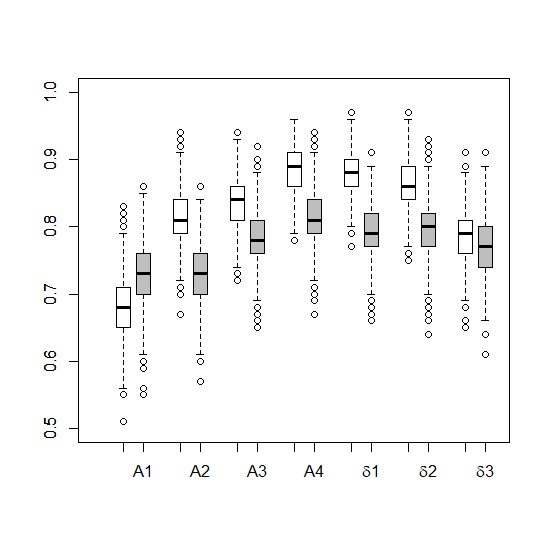}
\vskip-1cm
\caption{Boxplot for the total correct classification rates for all designs using nominal values and error standard deviations of $5 \times \hat \sigma$; white under $\eta_0$, grey under $\eta_1$.}
\label{boxplot}
\end{figure}
 
\section{Conclusions and possibilities of further research}\label{Sec:outlook}

We have presented a novel design criterion for symmetric model discrimination. Its main advantage is that design computations, unlike to $T$-optimality, can be undertaken with efficient routines of quadratic optimization that enhance the speed of computations by an order of magnitude. Also it was shown in an example that resulting designs are competitive in their actual discriminatory abilities. 

We have also introduced the notion of nominal confidence sets, which may have independent merit. 
Note again the distinction between parametric spaces and nominal confidence sets (and thus the principal distinction to `rigid' minimax approaches). 
Parametric spaces usually encompass all theoretically possible values of the parameters, while nominal confidence sets can contain the unknown parameters with very high probability, and still be significantly smaller than the original parameter spaces. 
 In this paper, we do not specify the process of constructing the nominal confidence regions, but if we perform a two stage experiment, with a second, discriminatory phase, the specification of the confidence sets is an important problem. 
\bigskip

As the approach suggested offers a fundamentally new way of constructing discriminatory designs, naturally many questions are yet unexplored and may warrant a closer look, see the following non-exhaustive list.

{\bigskip
\noindent\textbf{Sequential procedure.} 
The proposed method lends itself naturally to a two-stage procedure, where parameter estimates and confidence intervals are employed as nominal values in the second stage. Even sequential generation of design points can be straightforwardly implemented.
}   

\bigskip
\noindent\textbf{Approximate designs.} 
Proposition \ref{approx} is a possible gateway for the development of the standard approximate design theory for $\delta$-optimality, because the criterion $\delta^2_{app}$ is concave on the set of all approximate designs. Therefore, it is possible to work out a minimax-type equivalence theorem for $\delta$-optimal approximate designs, and use specific convex optimization methods to find a $\delta$-optimal approximate designs numerically. For instance, it would be possible to employ methods analogous to \cite{Burclova+P_16} or \cite{Yue+al_18}.
\bigskip

\bigskip


\noindent\textbf{Utilization of the $\delta$-optimal designs for related criteria.} For a design $\D=(x_1,\ldots,x_n)$, a natural criterion closely related to $\delta_r$-optimality can be defined as
\begin{eqnarray*}
 \tilde{\delta}_r(\D)=\inf_{\theta_0 \in \tilde\Theta_0^{(r)}, \theta_1 \in \tilde\Theta_1^{(r)}} \tilde{\delta}(\D |\theta_0,\theta_1), \text{ where }\\
 \tilde{\delta}(\D |\theta_0,\theta_1)=\left\|(\eta_0(\theta_0,x_i))_{i=1}^n-(\eta_1(\theta_1,x_i))_{i=1}^n\right\|.
\end{eqnarray*}
The criterion $\tilde{\delta}_r$ requires a multivariate non-convex optimization for the evaluation in each design $\D$, which entails possible numerical difficulties and a long time to compute an optimal design. However, the $\delta_r$-optimal design, which can be computed rapidly and reliably, can serve as efficient initial design for the optimization of $\tilde{\delta}_r$. Note that if $\tilde{\Theta}_0$ is a singleton containing only the nominal parameter value for Model 0, the $\delta_r$-optimal designs could potentially be used as efficient initial designs for computing the exact version of the criterion of $T$-optimality.
\bigskip

\noindent\textbf{Selection of the best design from a finite set of possible candidates.} 
As most proposals for the construction of optimal experimental designs, the method depends on the choice of some tuning parameters or even on entire prior distributions (in the Bayesian approach), which always results in a set of possible designs. It would be interesting to develop a comprehensive Monte-Carlo methodology for the choice of the best design out of this pre-selected small set of candidate designs. A useful generalization of the rule would take into account possibly unequal losses for the wrong classification.
\bigskip

{\noindent\textbf{Noncuboid sets.}
The methodology could certainly be extended to other types of confidence sets, particularly when we are interested in functional relations among the parameters . However then the particularly efficient box constrained quadratic programming algorithm could not be utilized.
\bigskip
}

\noindent\textbf{Higher-order approximations.}
As a referee remarked  it is possible to employ tighter approximations of the sets of mean values of responses than the one which we suggest. For instance, it would be possible to use the local curvature of the mean-value function. However, this may also lead to the loss of numerical efficiency of the method.
\bigskip

{
\noindent\textbf{More than two rival models.}
Another referee remark leads us to point out the natural extension to investigate a weighted sum or the minimum $\delta$ over all paired comparisons. The implications of this suggestions, however, requires deeper investigations.
}
\bigskip 

\noindent\textbf{Combination with other criteria.} The proposed method can produce {poor or even singular designs for estimating model parameters}. Because of this problem, which is btw. already mentioned in \cite{atkinson+f_75}, \cite{atkinson_08} used a compound criterion called $DT$-optimality. The same approach is possible for $\delta$-optimality. However, our numerical experience suggests that for a large enough size of the nominal confidence set, the delta-optimal designs tend to be supported on a set which is large enough for estimability of the parameters, without any combination with an auxiliary criterion. A detailed analysis goes beyond the scope of this paper.
\bigskip

\section*{Acknowledgements}

We are very grateful to Stefanie Biedermann from the University of Southampton for intensive discussions on earlier versions of the paper. We also thank Stephen Duffull from the University of Otago for sharing his code and Barbara Bogacka for sharing the data. Thanks to various participants of the design workshop in Banff, August 2017 and to Valerii Fedorov for many helpful comments. {We acknowledge the valuable inputs from four referees which lead to a considerable improvement of the paper.}

\end{document}